\documentclass[sigconf]{acmart}

\usepackage{booktabs} 
\usepackage{todonotes} 
\newtheorem{algorithm}{Algorithm}

\newlength\Colsep
\renewcommand\footnotetextcopyrightpermission[1]{} 
\begin{document}
\title[Defining Districts via Stable Matching]{Defining Equitable Geographic Districts in Road Networks via Stable Matching}


\author{David Eppstein}
\affiliation{%
  \institution{University of California, Irvine}
}
\email{eppstein@uci.edu}

\author{Michael T. Goodrich}
\affiliation{%
  \institution{University of California, Irvine}
}
\email{goodrich@uci.edu}

\author{Doruk Korkmaz}
\affiliation{%
  \institution{University of California, Irvine}
}
\email{dkorkmaz@uci.edu}

\author{Nil Mamano}
\affiliation{
  \institution{University of California, Irvine}
}
\email{nmamano@uci.edu}

\renewcommand{\shortauthors}{D. Eppstein et al.}

\begin{abstract}
We introduce a novel method for defining geographic districts in
road networks using stable matching.  In this approach, each
geographic district is defined in terms of a \emph{center}, which
identifies a location of interest, such as a post office or polling
place, and all other network vertices must be labeled with the
center to which they are associated.  We focus on defining geographic
districts that are \emph{equitable}, in that every district has the
same number of vertices and the assignment is stable in terms of
geographic distance.  That is, there is no unassigned vertex-center
pair such that both would prefer each other over their current assignments.
We solve this problem using a version of the classic stable matching
problem, called \emph{symmetric stable matching}, in which the
preferences of the elements in both sets obey a certain symmetry.
In our case, we study a graph-based version of stable matching in
which nodes are stably matched to a subset of nodes denoted as centers,
prioritized by their shortest-path distances, so that each center
is apportioned a certain number of nodes. We show that, for a planar
graph or road network with $n$ nodes and $k$ centers, the problem
can be solved in $O(n\sqrt{n}\log n)$ time, which improves upon the
$O(nk)$ runtime of using the classic Gale--Shapley stable matching algorithm
when $k$ is large. Finally, we provide experimental
results on road networks
for these algorithms and a heuristic algorithm that performs better
than the Gale--Shapley algorithm for any range of values of $k$.  
\end{abstract}

%
%



\maketitle


\section{Introduction}

\emph{Location analysis} is a classical branch of optimization in geographic information systems, concerned both with \emph{facility location}, the placement of centers to serve geographic regions such as polling places, fire stations, or post offices, and the \emph{assignment problem}, the problem of surrounding these facilities by service regions in an optimal way, so that all points are equitably served by nearby facilities and each facility bears a fair portion of the total service load.
This problem includes, for instance, the special case of 
\emph{political districting} in which the requirements for fairness (avoiding unfair gerrymandered districts) include both geographic compactness and 
equal representivity with respect to the broader population.
(E.g., see~\cite{niemi_deegan_1978,REVELLE20051,Ricca2008voronoi}.)

In this work, we consider a geographic abstraction of the assignment problem in which the facility locations have already been determined through some other 
algorithm. 
We model the geographic space of interest as a weighted, undirected graph representing a road network, we model the population to be assigned to facilities as the set of all vertices of the graph, and we model the facility locations as a subset of $k$ chosen \emph{center} nodes of the graph. Each center has a \emph{quota} indicating how many nodes it should match.
The desired output is an assignment of every node to a center, such that the set of assigned nodes for each center equals its quota. The use of quotas in this way allows each of the facilities to have different operational capacities in terms of how much of the population they can serve.

We impose the conditions that each node has a preference for centers 
ordered by shortest-path distance from the node, 
and each center has a preference for nodes 
ordered by their distances from the center.
Our goal is to match each center to its quota number of nodes and for the matching to be \emph{stable}, meaning that no node and center that are not assigned to each other prefer each other to their specified matches. Rather than optimizing some computationally challenging global quality criterion, we seek an assignment of nodes to centers that is stable. 

\begin{figure*}[hbt!]
\begin{tabular}{cc}
  \includegraphics[width=0.4\linewidth]{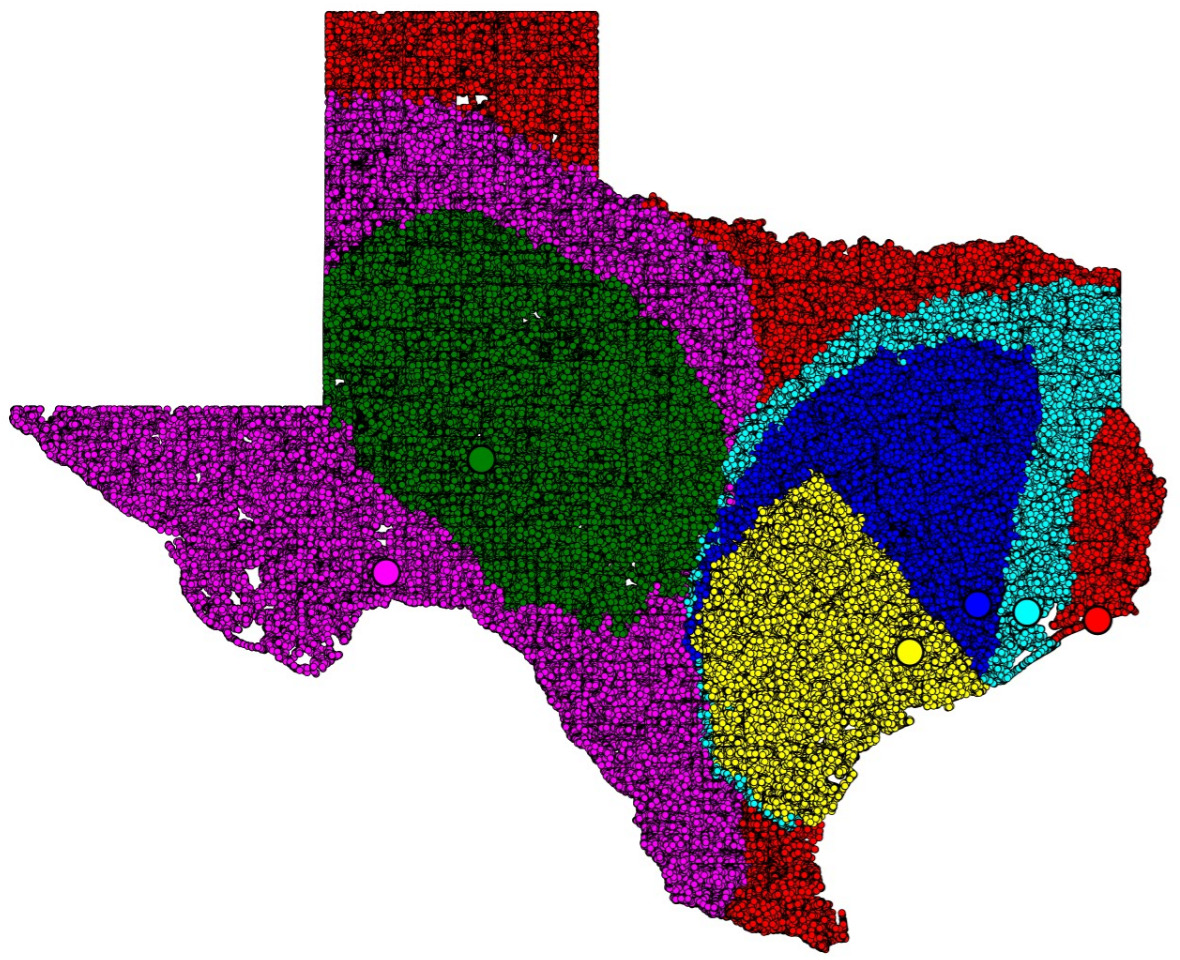} &   \includegraphics[width=0.4\linewidth]{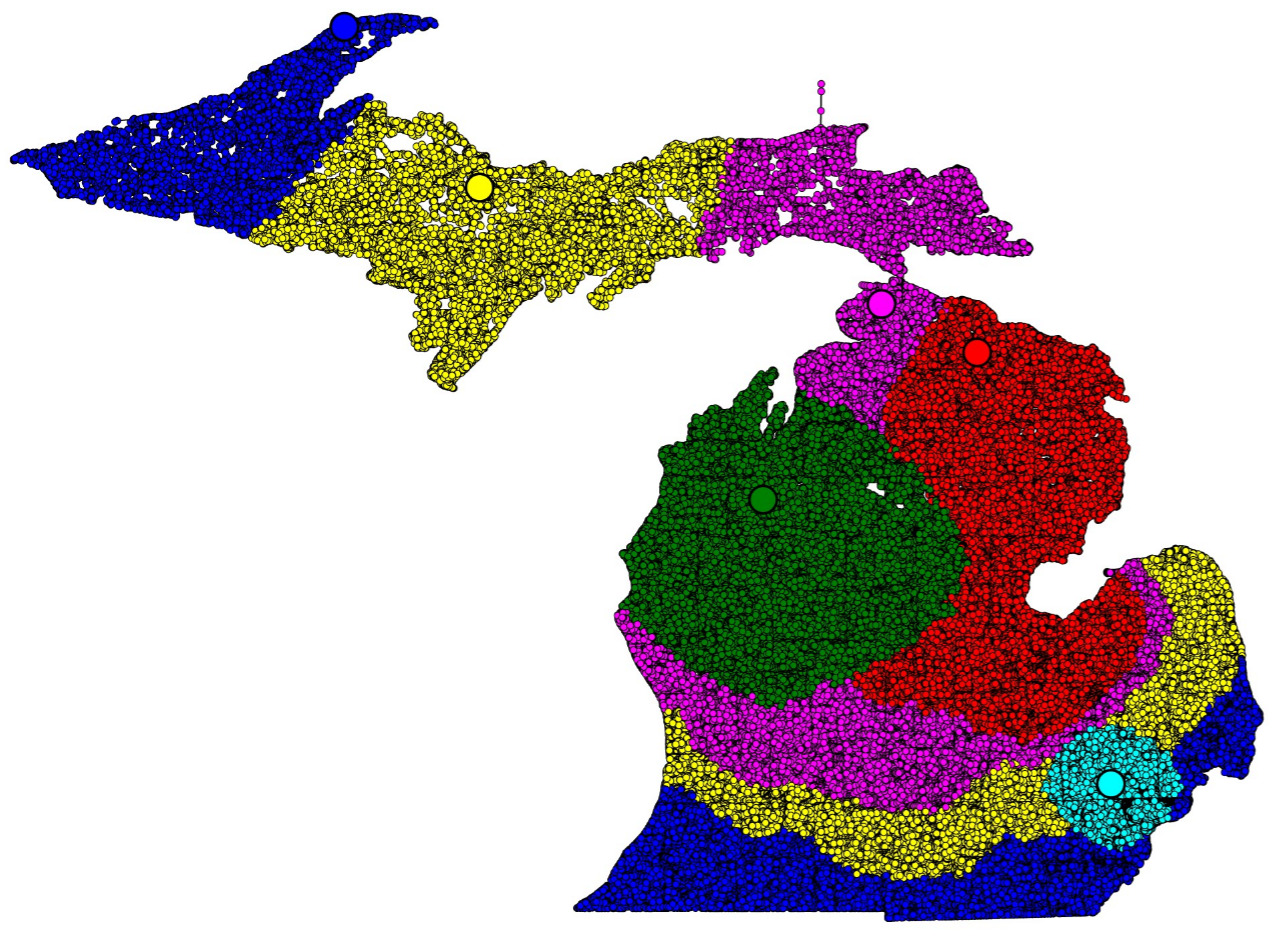} \\    Texas ($n=2037K,m=2550K$) & Michigan ($n=662K, m=833K$) \\[6pt]
 \includegraphics[width=0.45\linewidth]{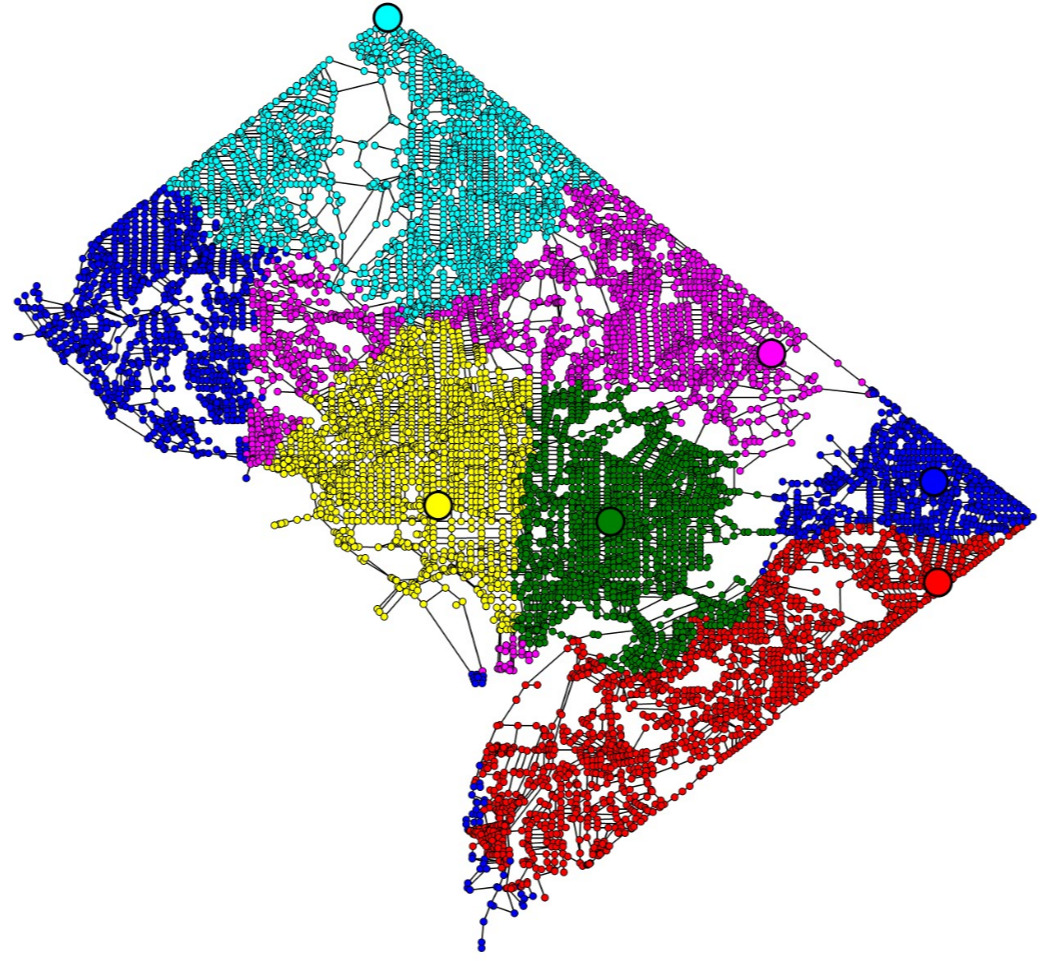} &   \includegraphics[width=0.28\linewidth]{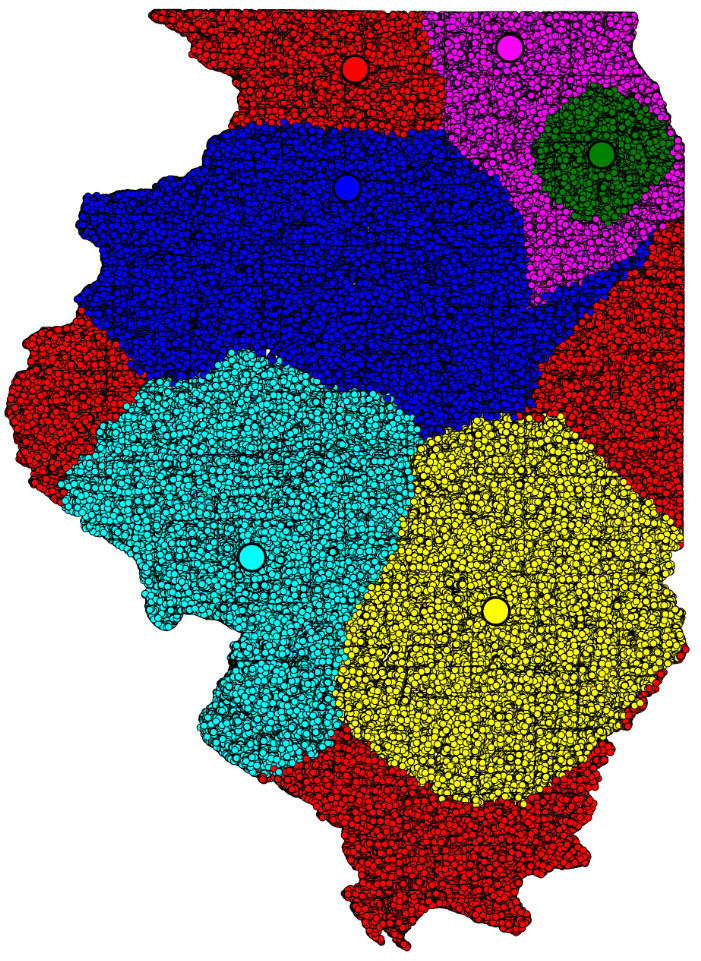} \\    Washington, DC ($n=9522,m=14850$) & Illinois ($n=790K,m=1008K$) \end{tabular}
\caption{The solutions to 
the \emph{stable graph matching} problem for the 2010 road networks of three 
U.S.~states and the District of Columbia, from the DIMACS database~\cite{DIMACS}. They consists of primary and secondary roads in the biggest connected component of the road networks. In each case, $n$ and $m$ denote the number of nodes and edges, respectively, and there are $k=6$ random centers with equal quota $n/k$.}
\label{fig:samplemaps}
\end{figure*}

We are offering this combined notion of giving each center a quota and 
optimizing stability in terms of distance-based preferences as a type of
\emph{equitability} for defining geographic districts. Quotas 
provide fairness in terms of the number of nodes assigned to each center and 
stability provides fairness in terms of how those nodes are assigned.

Defining geographic districts
that are equitable implies a certain amount of ``compactness'' for districts,
which avoids the types of highly non-compact districts that have been the
subject of recent legal cases involving gerrymandering.
This is a characteristic of our use of stable matching for assigning
nodes to centers based on symmetric distance-based preferences.
This notion does not, however,
imply that equitable districts are necessarily convex or even connected.
Indeed, depending on the placement of centers and how quotas 
are defined, it may be necessary for some districts to be disconnected.

Formally, we define the stable graph matching problem as follows:
\begin{definition}[Stable graph matching problem]
Given an undirected, weighted graph and a subset of $k$ nodes denoted centers, find an assignment from each node to a centers such that \textit{(i)} the same number of nodes is assigned to each center, up to round-off errors, and \textit{(ii)} the matching is stable with respect to shortest-path distances; that is, there is no node $u$ and center $c$ such that $u$ is not assigned to $c$, $u$ is closer to $c$ than to its assigned center, and $c$ is closer to $u$ than to one (any) of the nodes assigned to $c$. 
\end{definition}

Figure~\ref{fig:samplemaps} illustrates these properties for solutions to
our \emph{stable graph matching} problem for the 
road networks of three U.S.~states and the District of 
Columbia, with $k=6$ randomly-placed centers and equal quotas.

\subsection{New Results}
In the standard stable matching problem, preferences are arbitrary. 
Each individual may choose as his or her preferences any permutation of the opposite-set individuals, independently of all other choices. 
Preferences resulting from shortest-path distances in an undirected graph are 
not arbitrary, however. 
Instead, they obey a certain symmetry property 
coming from the undirected nature of the graph and shortest paths within the graph.
To capitalize on this idea, we define an abstract problem 
intermediate between stable graph matching and stable matching, which we call the
 \emph{symmetric stable matching problem}. 
Stable graph matching is a particular case of 
symmetric stable matching. We observe that for a symmetric stable matching of $n$ nodes with $k$ centers,  the classical Gale--Shapley algorithm can compute a solution in time $O(nk)$, once all distances between nodes and centers have been computed.

Moreover, we develop a novel \emph{nearest-neighbor chain algorithm} for any symmetric stable matching problem, using ideas borrowed from a very different application of nearest-neighbor chains, in hierarchical clustering problems~\cite{Ben-CAD-82,Jua-CAD-82}.
Our algorithm can be applied to stable graph matching and extends our previous work on \emph{stable grid matching}~\cite{EPPSTEIN2017}, another case of symmetric stable matching.
It runs in $O(n\cdot T(n))$ time, where $T(n)$ is the time per operation of a data structure
for updating a dynamic subset of points from the given metric space and answering nearest neighbors to these points. In the graph setting, this means that we need to be able to find the closest center of a node, and vice versa, efficiently. By using the data structure for road networks from~\cite{EPPSTEIN20173}, with $T(n)=O(\sqrt{n}\log n)$, the stable graph matching problem can be solved in $O(n\sqrt{n}\log n)$ time.

We summarize our contributions as follows.
\begin{itemize}
\item We formulate the stable graph matching problem and its generalization, the symmetric stable matching problem.
\item We describe a general class of algorithms for solving symmetric stable matching, the mutual closest pair algorithms, and prove that for symmetric stable matching (and unlike stable matching more generally) the solution is always unique.
\item We define and analyze the nearest-neighbor chain algorithm for symmetric
stable matchings. As we show, for inputs that can support nearest-neighbor
queries on dynamic subsets of the input preferences, with time $T(n)$ per update
or nearest-neighbor query, we can find a symmetric stable matching in time
$O(n\, T(n))$.
\item We provide a heuristic circle-growing improvement to the Gale--Shapley algorithm for the case of stable graph matching. 
Our heuristic does not improve the $O(nk)$ worst-case time of the algorithm, but we expect it to provide significant speedups in practice.
\item We provide an experimental comparison of our algorithms on real-world road networks. 
Our experiments confirm the independence from $k$ of the running time of our nearest-neighbor based algorithm, and they also confirm the efficacy of our heuristic 
circle-growing improvement to the Gale--Shapley algorithm.
\end{itemize}

\subsection{Prior Related Work}
Our notion of equitability introduces 
an interesting new (and more realistic)
twist to Knuth's classic post office problem~\cite{knuth1998art}.
In the classic post office problem, one is given a collection of sites
called ``post offices'' and one is interested in assigning
nodes to their nearest post office with no consideration for quotas
characterizing the capacity of each post office to handle mail.
Thus, the classic post office problem is equivalent to
our geographic districting problem with unbounded quotas.
Knuth's discussion of the classic post office problem has given rise to 
a long line of research on spatial partitioning, including the important Voronoi diagrams (e.g., see~\cite{Aurenhammer:1991}),
which have also been extended to the graph setting~\cite{Erwig2000}.

The \emph{stable matching problem}, which is also known
as the \emph{stable marriage problem}, 
was introduced by Gale and Shapley~\cite{gale62}.
This problem was 
originally described in terms of matching $n$ men and $n$ women 
based on each person having an ordered preference list for the members
of the opposite sex in this group.
In that context, stability means that no man--woman pair prefer 
each other to their assigned choices. 
Stability, defined in this way, is a necessary condition (and more important than, e.g., total utility) in order to prevent extramarital affairs.
When generalized to the one-to-many case, this problem
is also called the \emph{college admission problem}~\cite{Roth89}, 
because it models a setting where $n$ students are 
stably matched to $k<n$ colleges, each with a certain quota of admissions.
Indeed, solutions to this one-to-many stable matching
problem are currently used to match medical students to residency programs in some countries, such as the US.

The standard algorithm used today
for computing a stable matching with arbitrary preferences 
is the original Gale--Shapley algorithm~\cite{gale62}.
This algorithm finds a stable matching in time $O(nk)$ 
in the one-to-many (college admission) case or in $O(n^2)$ time in the one-to-one 
(men and women) case. 
Moreover, these running 
times are the best possible for arbitrary preferences, since just
reading the input in the one-to-many case requires $\Omega(nk)$ time
and $\Omega(n^2)$ time in the one-to-one case.
Intuitively, the one-to-one version of their algorithm 
involves each man making proposals to women according to his 
preference order and each
woman accepting a proposal if it is her first or if it offers her a better
match based on her preference order.

Existing research about stable matching studies variations such as matching with added constraints~\cite{KOJIMA2010}, preferences with ties~\cite{IRVING1994}, and many more, e.g., see~\cite{IWAMA2008,MANLOVE2002}. 
However, the assumption that preferences are arbitrary has rarely been challenged. A first step in this direction was taken by~\citet{hoffman2006}, who considered the mathematical properties of a stable matching in a geometric setting, where ``colleges'' are points in $\mathbb{R}^2$ and ``students'' are all the points in $\mathbb{R}^2$, and both use distances as preferences. 
Eppstein {\it et al.}~\cite{EPPSTEIN2017} extended their approach to images, where
``students'' and ``colleges'' are pixels, but their work does not extend to
general graphs and road networks.

\section{Symmetric stable matching}\label{sec:ssm}
We present the symmetric stable matching problem in the one-to-many context of schools and students, and therefore all the results in this section also apply to the one-to-one case of men and women.

In order to formulate the symmetric stable matching problem, consider this alternative but equivalent definition of the stable matching problem. Each agent (school or student) gives a unique score to each agent from the other set, and ranks them in increasing order of these scores. Therefore, a set of scores such as $(a\gets 7, b\gets 2, c\gets 10)$ corresponds to the list of preferences $(b,a,c)$.

We call the preferences \emph{symmetric} if the score of $x$ for $y$ equals the score of $y$ for $x$. Moreover, in this case, we call these scores \emph{distances}.

\begin{definition}
A stable matching problem is \emph{symmetric} if the preferences are symmetric.
\end{definition}

\subsection{Mutual closest pair algorithm}

Before introducing our nearest-neighbor chain algorithm for symmetric stable matching, we describe a simplified version of it, the mutual closest pair algorithm.

\begin{definition}
In a stable matching problem, a \emph{mutual closest pair} is a school and a student who have each other as first choice.
\end{definition}

The algorithm is based on the following lemma:

\begin{lemma}\label{lem:mutual}
If preferences are symmetric, a mutual closest pair always exists.
\end{lemma}

\begin{proof}
Let $s$ and $x$ be the student and school whose distance is the global minimum, that is,
no other school and student are closer to each other than $s$ and $x$.
Then, $s$ and $x$ are a mutual closest pair: $s$ is the closest student to $x$, and $x$ is the closest school to $s$.
\end{proof}

Although the pair realizing the global minimum distance are always a mutual closest pair, the reverse is not true: there can be other mutual closest pairs whose distance is not a global minimum.
Moreover, Lemma~\ref{lem:mutual} and its proof require that the distances on which we are basing preferences be symmetric. If they are not symmetric, as may be the case for shortest path distances in a directed graph, then there might not be any mutual closest pairs.

Now we describe our symmetric stable matching algorithm:
\begin{algorithm}\label{alg:mutual}\normalfont{Mutual closest pair algorithm}
\begin{description}
\item[Input:] $n$ students and $m$ schools with symmetric preferences, and school quotas adding up to $n$.
\item[Output:] a stable matching between the students and schools.
\end{description}
\begin{enumerate}
\item Initialize the matching empty.
\item Repeat while there is an unmatched student:
\begin{enumerate}
\item Find a mutual closest pair $s,x$.
\item Match $s$ and $x$, remove the student from the pool of unmatched students, reduce the quota of the school by one and remove it from the pool of unmatched schools if its the quota reached zero.
\end{enumerate}
\end{enumerate}
\end{algorithm}
Due to Lemma~\ref{lem:mutual}, the algorithm will never fail to find a closest mutual pair. Next, we prove that the resulting matching is stable, that is, that there are no blocking pairs (a student and a school that are not matched to each other but prefer each other to their assigned choices).

\begin{theorem}
\label{thm:alg-mutual}
Algorithm~\ref{alg:mutual} finds a stable solution to any symmetric stable matching problem.
\end{theorem}
\begin{proof}
Suppose student $s$ is matched to school $x$, but she prefers school $y$. When $s$ and $x$ were matched by the algorithm, $x$ was the closest school to $s$, so $y$ already had quota zero. Therefore, $y$ matched all of its students while $s$ was not matched yet. But $y$ was matched with students other than $s$, who must have been closer to $y$ than $s$. Hence, $s$ and $y$ are not a blocking pair.
\end{proof}

\begin{lemma}
\label{lem:unique}
If $s$ and $x$ are a mutual closest pair of a symmetric stable matching problem,
then every stable solution to the problem must match $s$ to $x$.
\end{lemma}

\begin{proof}
If $s$ and $x$ were each matched to someone other than each other, they would form an unstable pair.
\end{proof}

\begin{theorem}
Any symmetric stable matching problem has a unique solution, which will be found by any instance of Algorithm~\ref{alg:mutual}.
\end{theorem}

\begin{proof}
Let $S$ be any solution to the problem; at least one solution $S$ exists by \autoref{thm:alg-mutual}.
Let $s$ and $x$ be a mutual closest pair, which must exist by Lemma~\ref{lem:mutual}.
By Lemma~\ref{lem:unique}, $s$ and $x$ are matched to each other in $S$.
Because Algorithm~\ref{alg:mutual} is guaranteed to find a correct solution (\autoref{thm:alg-mutual} again), it necessarily matches $s$ with $x$, so its behavior on $s$ and $x$ agrees with solution $S$. These two elements $s$ and $x$ cannot form any mutual closest pairs with other elements, so if we remove both of them from the given problem, we obtain a smaller problem such that the restriction of $S$ to the smaller problem is still stable and such that the restriction of Algorithm~\ref{alg:mutual} to the smaller problem agrees with its behavior on the whole problem. The result follows by induction on the size of the problem. 
\end{proof}

Eeckhout~\cite{Eeckhout2001} stated a sufficient condition for a unique solution in one-to-one stable matchings. It can be shown that symmetric preferences satisfy this condition, and hence uniqueness also follows from their result in the one-to-one case.

Algorithm~\ref{alg:mutual} leaves open how to actually find a mutual closest pair, but any strategy that finds mutual closest pairs will work correctly. Although different strategies may choose the matches that they make in different orderings, and may take different running times, they will always produce the same overall matching. In the next section, we present one strategy for quickly finding mutual closest pairs by making use of a dynamic nearest-neighbor data structure. This data structure should be able to maintain a set of agents of the same type (students or schools) and answer queries asking for the closest one to a query agent of the opposite set. Moreover, it should support deletions, that is, allow to remove elements from the set.

\subsection{Nearest-neighbor chain algorithm}

The following algorithm, which we call the Nearest-neighbor chain algorithm, is based on the theory of hierarchical clustering~\cite{Ben-CAD-82,Jua-CAD-82}, and was first used in the context of stable matching (for grid-based geometric data only) in~\cite{EPPSTEIN2017}.

\begin{algorithm}\label{alg:chain}\normalfont{Nearest-neighbor chain algorithm}
\begin{description}
\item[Input:] $n$ students and $m$ schools with symmetric preferences, and school quotas adding up to $n$.
\item[Output:] a stable matching between the students and schools.
\end{description}
\begin{enumerate}
\item Initialize the matching empty.
\item Initialize a dynamic nearest-neighbor structure containing the students, and one containing the schools.
\item Initialize an empty stack $S$.
\item Repeat while there is an unmatched student:
\begin{enumerate}
\item If $S$ is empty, add any unmatched student (or school) to it.
\item Let $p$ be the agent at the top of the stack, and use the nearest-neighbor structures to find its nearest-neighbor $q$ of the opposite set.
\item If $q$ is not already in $S$, add it.
\item Otherwise, $q$ must be the second-from-top element in $S$ (as justified below), and $p$ and $q$ are a mutual closest pair. In this case, match $p$ and $q$, and update the data structures accordingly: remove the student from the nearest-neighbor structure of students, reduce the quota of the school by one and remove it from the nearest-neighbor structure of schools if its quota reached zero, and remove $p$ and $q$ from the stack. Note that if the school was below the student in the stack and it still had positive quota, it would be added to the stack again in the next iteration, as it would still be the nearest-neighbor of the previous student in the stack. Hence, in this case, we can keep the school in the stack.
\end{enumerate}
\end{enumerate}
\end{algorithm}
Note that the distance between consecutive elements in $S$ only decreases. That's why, in Step (4d), $q$ must be the second-from-top; if $q$ was anywhere else, $p$ would be closer to its predecessor in $S$ than to $q$. Here we are using the fact that the preferences of each element are distinct. In the graph setting, we may use a tie-breaking rule to ensure that distances are unique.

Each step that adds a new element to $S$ can be charged against a later pop operation and its associated match. Therefore, the number of repetitions is $O(n)$. This algorithm gives us the following theorem.

\begin{theorem}\label{the:chain}
The symmetric stable matching problem can be solved in $O(n)$ query and update 
operations of a dynamic nearest-neighbor data structure.
\end{theorem}

By combining Theorem~\ref{the:chain} with the data structure for planar graphs and road networks from~\cite{EPPSTEIN20173}, the stable graph matching problem can be solved in $O(n\sqrt{n}\log n)$ time.

\subsection{Circle-growing algorithm}

As we mentioned in the introduction, the Gale--Shapley algorithm requires $O(nk)$ time to find a stable matching between $n$ nodes and $k$ centers. However, in the graph setting first we need to compute the preferences, that is, the shortest-path distances between every center and node. These preferences can be computed by applying a single-source shortest-path algorithm starting from each center, such as Dijkstra's algorithm. The running time for applying Dijkstra's algorithm $k$ times on a graph with $m$ edges and $n$ vertices, using a Fibonacci heap based implementation of Dijkstra's algorithm, would be $O(k(m+n\log n))$, e.g., see~\cite{Cormen2001}. In planar graphs, this could be improved to $O(kn)$ by replacing Dijkstra's algorithm with the linear-time algorithm from Henzinger {\it et al.}~\cite{HENZINGER19973}. Therefore, the time for this preference computation step matches or dominates the time for performing the Gale--Shapley algorithm.

However, with the following alternative algorithm it is not necessary to compute the distances between all centers and nodes, and we can do without a separate Gale--Shapley phase of the algorithm altogether. Instead, we perform the assignment steps of the solution as part of $k$ instances of Dijkstra's algorithm, allowing us to stop each instance earlier once its quota is met.
 The algorithm is analogous to the \emph{circle-growing method}, a geometric algorithm for a continuous variant of stable grid matching described by Hoffman
{\it et al.}~\cite{hoffman2006}. It can be visualized as a process in which we grow circles from each center, all at the same speed, and match each node to the first circle that grows across it.

We start $k$ instances of Dijkstra's algorithm at the same time, one from each center. We explore, at each step, the next closest node to any of the centers, advancing one of the instances of Dijkstra's algorithm by a single step. We match each node to the center whose instance of Dijkstra's algorithm reaches it first. Note that when an instance of Dijkstra's algorithm, starting from center $c$, reaches a node $x$ that has not already been matched, then $c$ and $x$ must be the global closest pair (omitting already matched pairs). We halt each instance of Dijkstra's algorithm as soon as its center reaches its quota. This stopping condition prevents wasted work in which an instance of Dijkstra's algorithm explores nodes farther than its farthest matched node. In addition, using this method to solve symmetric stable matching problems allows us to avoid running the Gale--Shapley algorithm afterwards.

There are several alternatives for implementing the parallel instances of Dijkstra's algorithm, all of which result in a running time of $O(k(m+n\log n))$ for arbitrary graphs and $O(kn\log n)$ for planar graphs. For instance, we can use a priority queue of centers to decide which instance of Dijkstra's algorithm should advance in each step, or we can merge the priority queues of all the instances of Dijkstra's algorithm into a single larger priority queue, which is then implemented
using a Fibonacci heap.

\section{Experiments}
In this section we present an empirical comparison of the Gale--Shapley algorithm, circle-growing algorithm, and nearest-neighbor chain algorithm on real-world road network data.
Figure~\ref{fig:runtimeplots} and its associated Table~\ref{tab:DE} illustrate the main findings.

\begin{figure*}[t]
  \centering
  \includegraphics[width=0.85\linewidth]{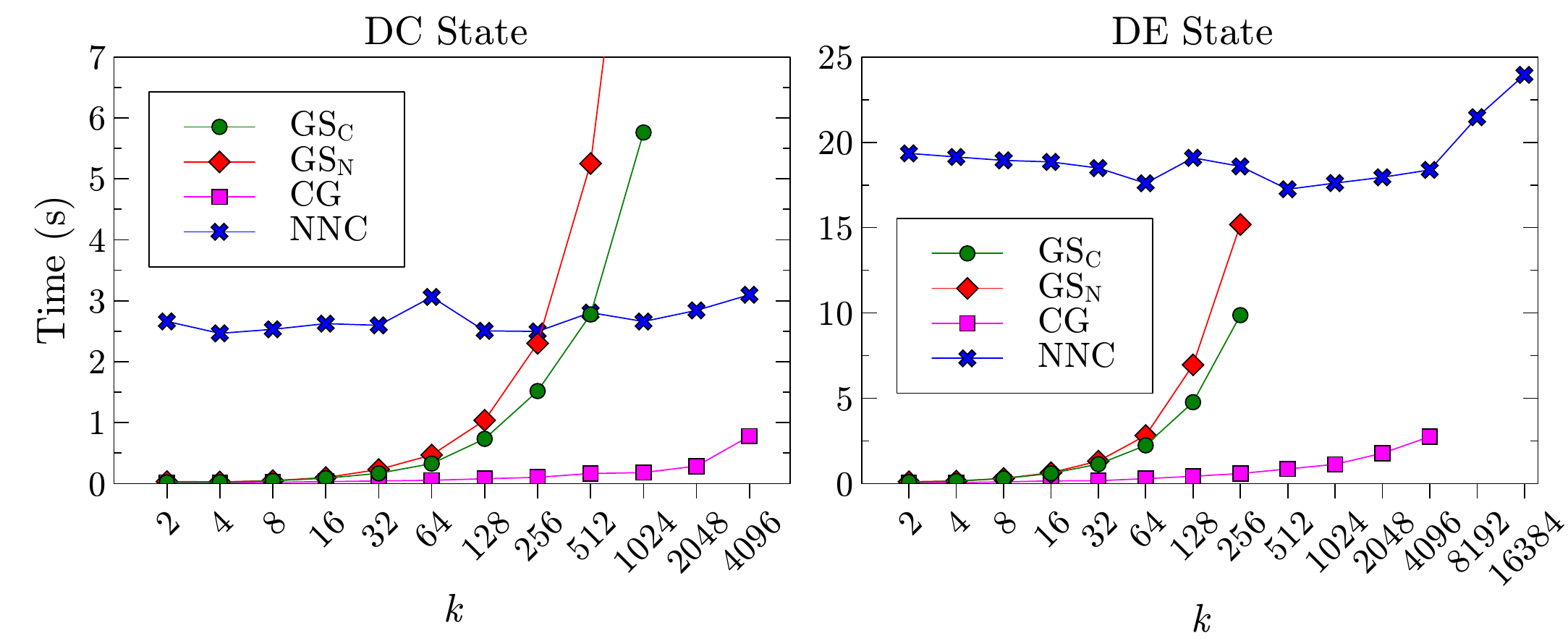} 
  \caption{Comparison of the running time of the algorithms in the Washington, DC (left, $n=9522,m=14850$) and Delaware (right, $n=48812,m=60027$) road networks from the DIMACS database~\cite{DIMACS} for a range of number of centers $k$ (in a logarithmic scale). Each data point is the average of 10 runs with 10 sets of random centers (the same sets for all the algorithms).}
  \label{fig:runtimeplots}
\end{figure*}

\subsection{Experiment setup}

We implemented the various symmetric stable matching algorithms of our comparison in Java~8.
We then executed them and timed them as run on an Intel Core CPU i7-3537U 2.00GHz with 4GB of RAM, under Windows 10.

In the table and figures presenting our experimental results, we use the label $CG$ for the circle-growing algorithm and $NNC$ for the nearest-neighbor chain algorithm. For Gale--Shapley, we consider a variation $GS_C$ where the centers do the proposals (which corresponds to the role of the men in the original algorithm), and the alternative $GS_N$ where the nodes do the proposals. For the nearest-neighbor chain algorithm, we implemented and used the dynamic nearest-neighbor data structure from~\cite{EPPSTEIN20173}.

\begin{table}[]
  \centering
  \caption{Runtime in seconds of the algorithms in the Delaware road network ($n=48812,m=60027$). Each data point is the average of 10 runs with 10 sets of random centers (the same sets for all the algorithms). A dash indicates that the algorithm ran out of memory.}
  \label{tab:DE}
\begin{tabular}{lllll}
  $k$    & $GS_N$ & $GS_C$ & $CG$   & $NNC$   \\
  2    & 0.11  & 0.09  & 0.06 & 19.35 \\
  4    & 0.15  & 0.15  & 0.06 & 19.14 \\
  8    & 0.30  & 0.30  & 0.10 & 18.94 \\
  16   & 0.64  & 0.60  & 0.16 & 18.85 \\
  32   & 1.32  & 1.14  & 0.17 & 18.49 \\
  64   & 2.82  & 2.24  & 0.29 & 17.60 \\
  128  & 6.96  & 4.77  & 0.43 & 19.09 \\
  256  & 15.18 & 9.87  & 0.59 & 18.59 \\
  512  & ---    & ---    & 0.86 & 17.25 \\
  1024 & ---    & ---    & 1.13 & 17.61 \\
  2048 & ---    & ---    & 1.78 & 17.95 \\
  4096 & ---    & ---    & 2.75 & 18.38 \\
  8192 & ---    & ---    & --- & 21.47 \\
  16384 & ---    & ---    & --- & 23.95 \\
\end{tabular}
\end{table}
  
\subsection{Results}

Figure~\ref{fig:runtimeplots} shows a clear picture of the respective algorithms' strengths and weaknesses:
\begin{itemize}
\item The Gale--Shapley algorithm, with a runtime of $O(k n\log n)$, scales linearly with $k$. Moreover, because of the memory requirement of $\Theta(nk)$, we could not run it with large numbers of centers. The version of Gale--Shapley where nodes propose ($GS_N$) was about $50\%$ slower than the version where centers propose. This is explained by the fact that, when nodes propose, each center needs to keep track of its least preferred already-matched node. This node may need to be rejected if the center receives a preferable proposition from another node. We maintain these least-preferred matched nodes by using a binary heap of nodes for each center; however, the overhead of maintaining this heap adds to the running time of our implementation.
In contrast, when centers propose, each node needs to keep track only of a single match, so we do not need to use an additional binary heap for this purpose.

\smallskip
\item Our circle-growing algorithm was the fastest of our implemented algorithms
in practice, over the range of values of $k$ for which we could run it. It is also the only algorithm that could complete a solution for the largest road networks that we tested. For instance, on the Texas road network, which has over 2 million nodes, the algorithm finishes in 3 seconds when given 6 random centers; our other implementations could not solve instances this large.
We did not see significant differences in the runtime between different ways to implement the parallel instances of Dijkstra's algorithm.

\smallskip
\item Additionally, in contrast to the Gale--Shapley algorithm, the runtime of circle-growing did not appear to be strongly affected by the value of~$k$. The reason for this is that, even though the algorithm runs $k$ instances of Dijkstra's algorithm, the expected number of nodes that each instance explores decreases as $k$ increases. However, this phenomenon may only be valid in expectation with randomly located centers.

\smallskip
\item Our nearest-neighbor chain algorithm, with a runtime of $O(n\sqrt{n}\log n)$, is the only one with a runtime independent of $k$. Hence, it has a flat curve in the plots\footnote{Even though the runtime of NNC seems to start to increase for the largest values of $k$ in the Delaware plot, this is likely to be due to an unrelated hardware/software issue, as other simulations of this comparison did not show this.}. The Gale--Shapley curve and the nearest-neighbor chain curve cross in our experimental data at around $k\approx 4\sqrt{n}$, showing that the constant factors in our implementation of the nearest-neighbor chain algorithm are reasonable. Moreover, because of its memory requirement of $O(n\sqrt{n})$, the nearest-neighbor chain algorithm is the only algorithm that was able to complete a solution for the entire range of values of $k$ on all inputs that were small enough for it to run at all. For instance, in the Delaware road network, the Gale--Shapley algorithm ran out of memory at $k=256$, and the circle-growing algorithm ran out of memory at $k=8192$, but the nearest-neighbor chain algorithm was unaffected by the choice of $k$.
\end{itemize}

\section{Conclusions}
We have defined the symmetric stable matching problem, a subfamily of stable matching problems which arise naturally when preferences are determined by distances. We studied its basic properties and provided the \emph{mutual closest pair algorithm}, which has the potential to be faster than the Gale--Shapley algorithm. Future researchers should consider the algorithms in this paper if they identify that a matching problem has symmetric preferences. As a special case of symmetric stable matching, we defined the stable graph matching problem. For this problem, we compared (a) the Gale--Shapley algorithm, (b) the mutual closest pair algorithm, and (c) the \emph{circle-growing} algorithm, a heuristic improvement over the Gale--Shapley algorithm.

This work leaves open several  questions for future research:
\begin{itemize}
\item We know of two settings where symmetric stable matching arises naturally: the geometric and graph-based cases. In what other settings does it arise?
\item The experiments show that the circle-growing algorithm scales much better with $k$ than the Gale--Shapley algorithm when centers are placed randomly. However, in the worst case both are $\theta(kn\log n)$. For instance, the circle-growing algorithm achieves this worst-case behavior when the graph is just a path and the $k$ centers are located at the first $k$ nodes.
Can the circle-growing algorithm be shown to have a better expected complexity when centers are placed randomly?
\item As Figure~\ref{fig:samplemaps} illustrates, the regions given by the solution to an instance of the stable graph matching problem are not necessarily connected. In many applications, such as political districting~\cite{Ricca2008voronoi},  it is necessary to have connected regions. Because stable graph matching has a unique solution, achieving connectivity will require a relaxation of the stability property, or a better-optimized choice of the center locations. We leave these for future work.
\end{itemize}

\begin{acks}
This article reports on work supported by the
DARPA under agreement no.~AFRL FA8750-15-2-0092.
The views expressed are those of the authors and do not reflect the
official policy or position of the Department of Defense
or the U.S.~Government.
This work was also supported in part from NSF grants
1228639, 1526631,
1217322, 1618301, and 1616248.
\end{acks}

\newpage 
\bibliographystyle{ACM-Reference-Format}
\bibliography{biblio} 


\appendix

\section{Additional results on road networks}
In this appendix we visualize additional examples of stable graph matchings in real road networks. 

Figure~\ref{fig:samplemaps2arxiv} shows the results for 6 additional US states.
Figure~\ref{fig:diffkarxiv} shows the same road network with different numbers of centers.

\begin{figure*}
\begin{tabular}{cc}
\includegraphics[width=0.6\linewidth]{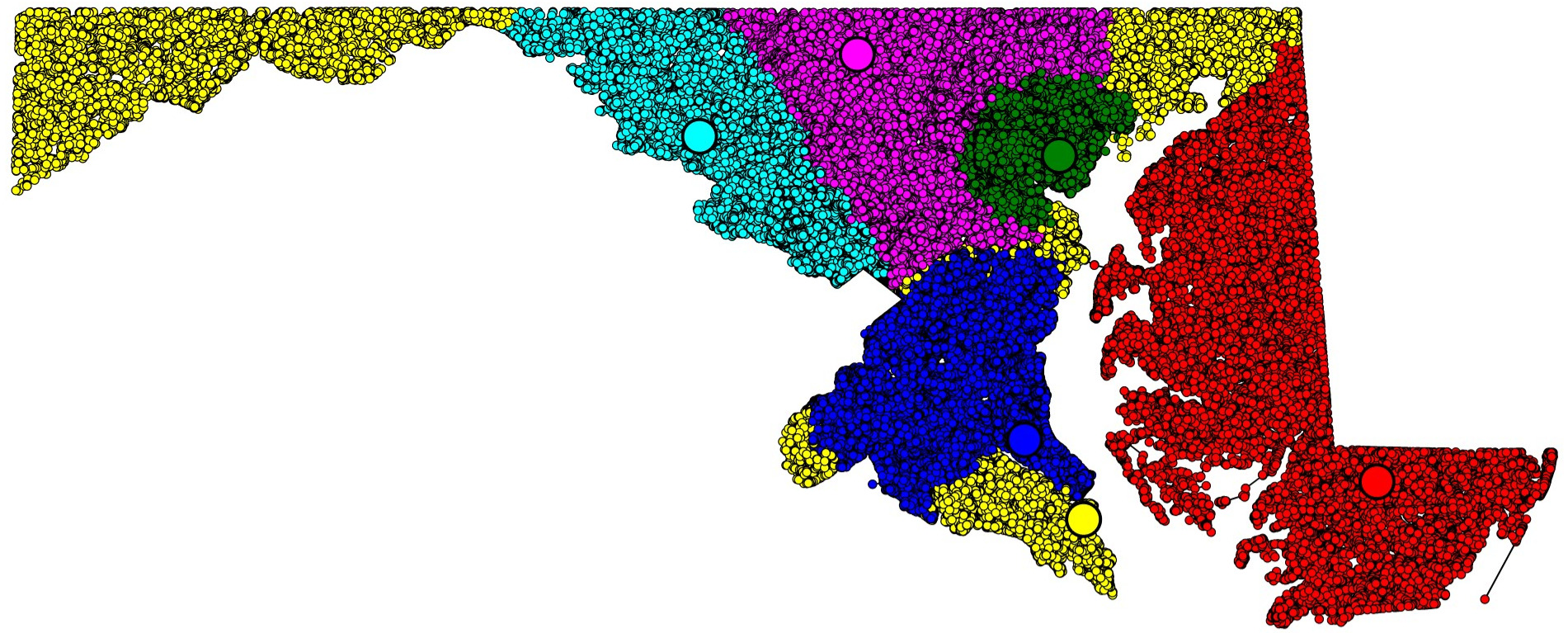} &   \includegraphics[width=0.4\linewidth]{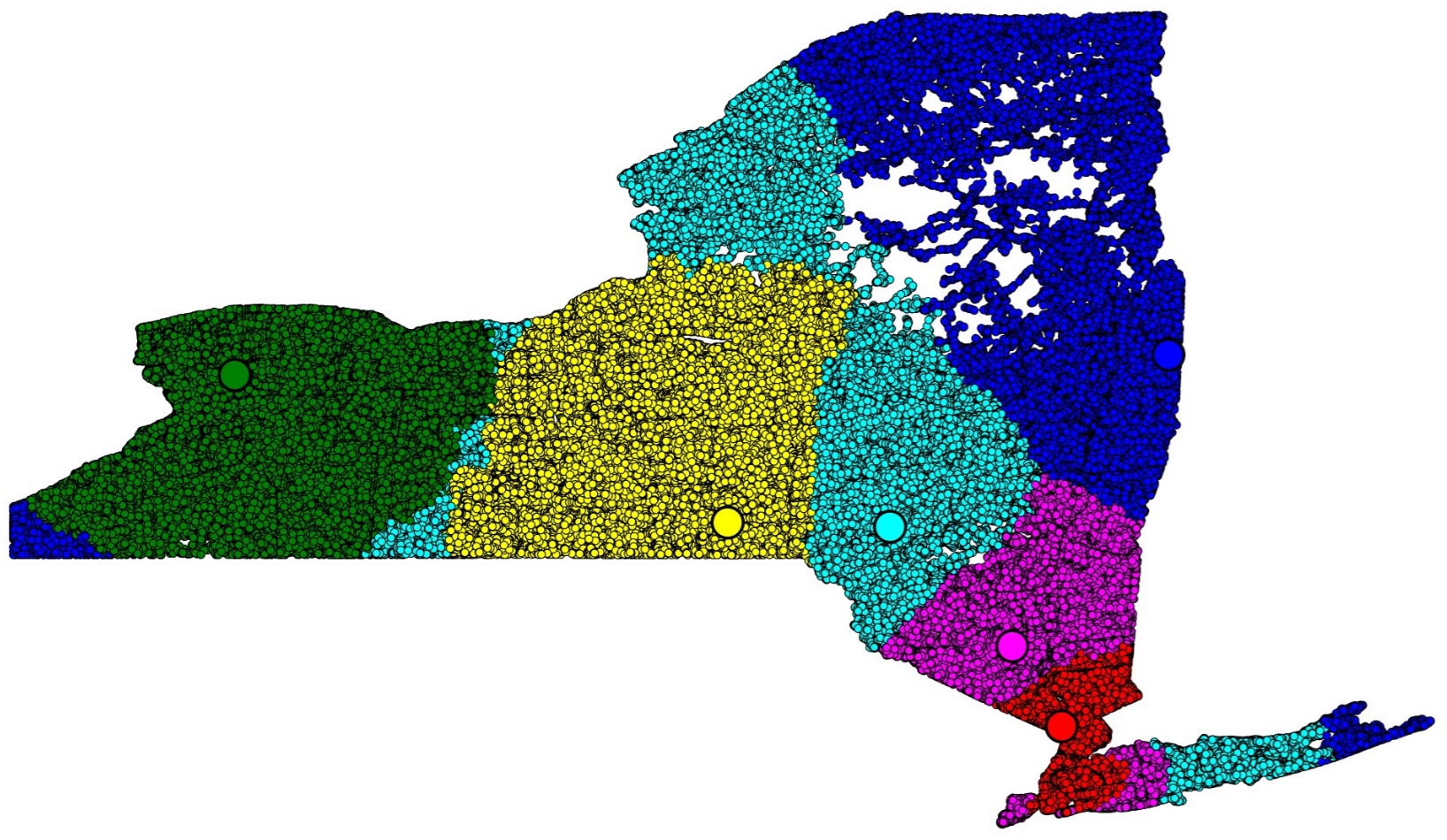} \\
Maryland ($n=264K,m=315K$) & New York ($n=709K, m=889K$)\\[6pt]
\includegraphics[width=0.45\linewidth]{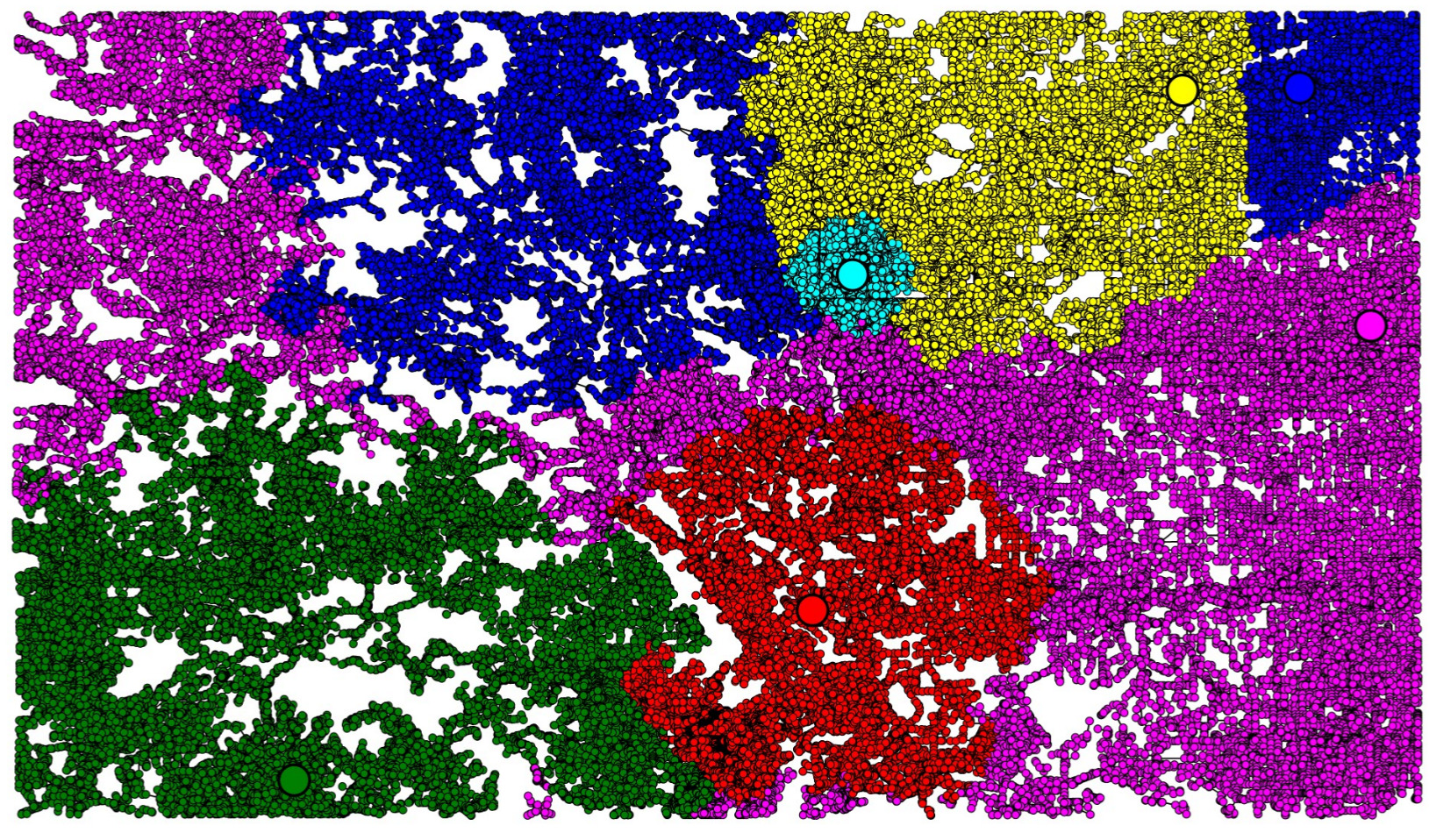} & \includegraphics[width=0.4\linewidth]{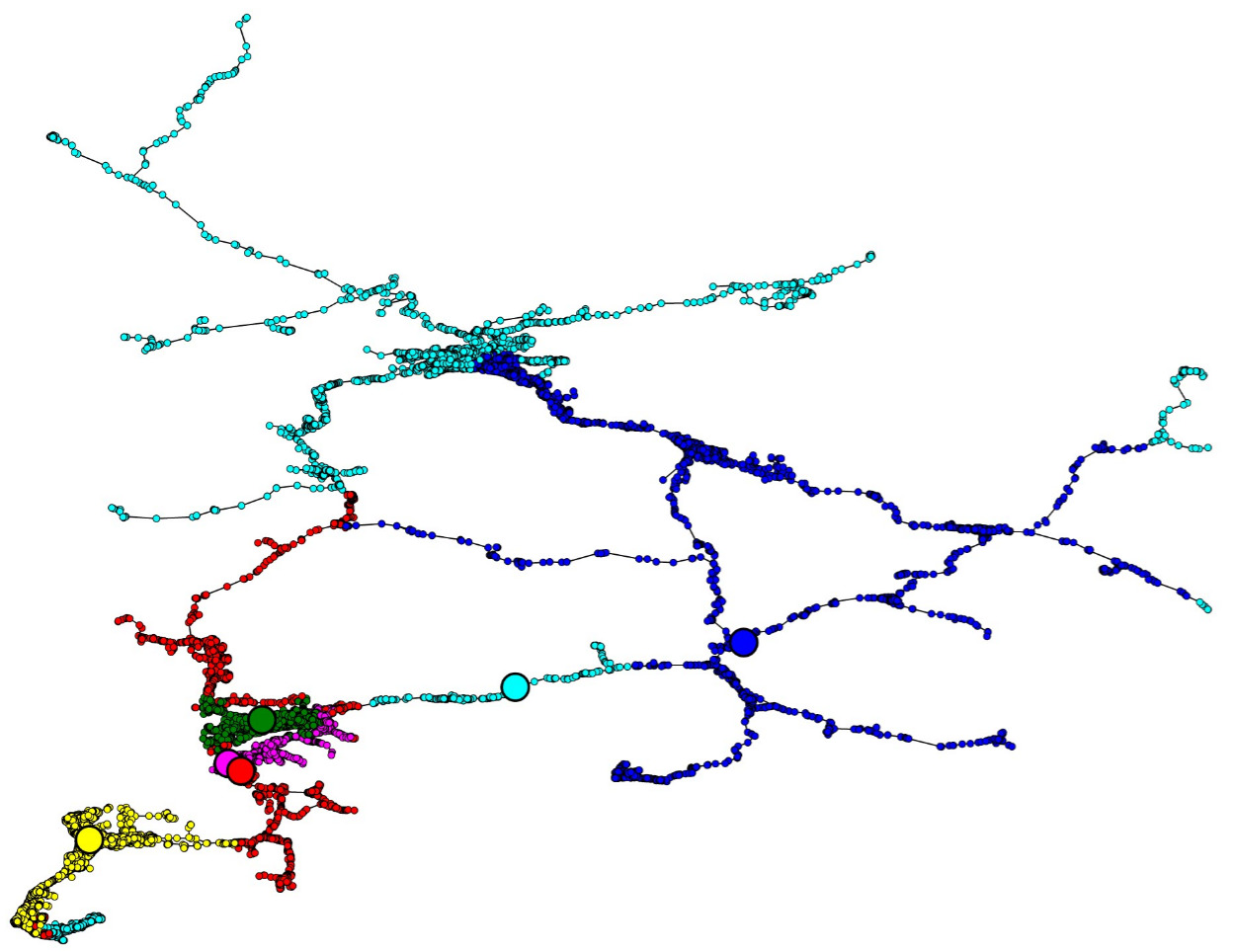} \\
Colorado ($n=436K,m=529K$) & Alaska ($n=49K, m=55K$) \\[6pt]
\includegraphics[width=0.45\linewidth]{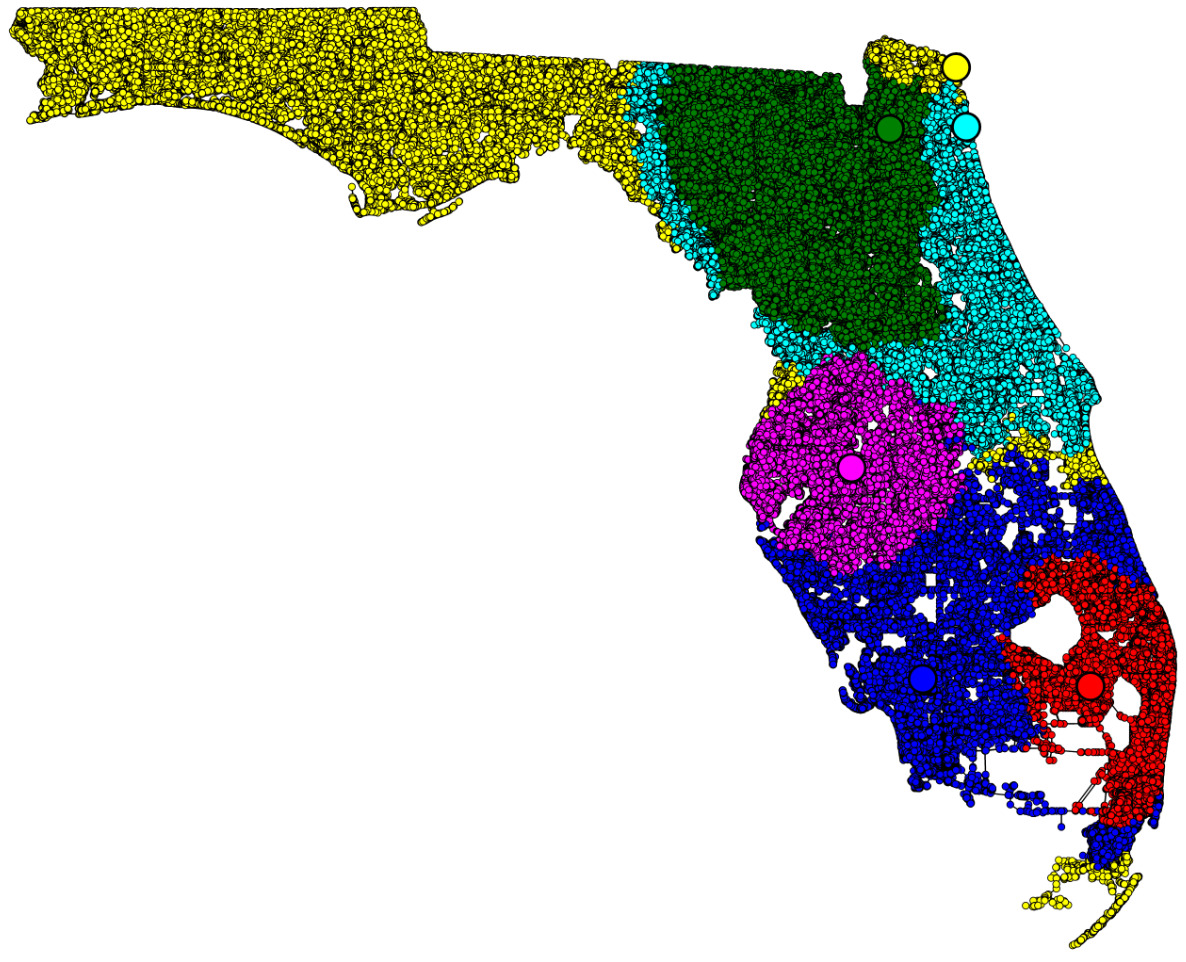} & \includegraphics[width=0.2\linewidth]{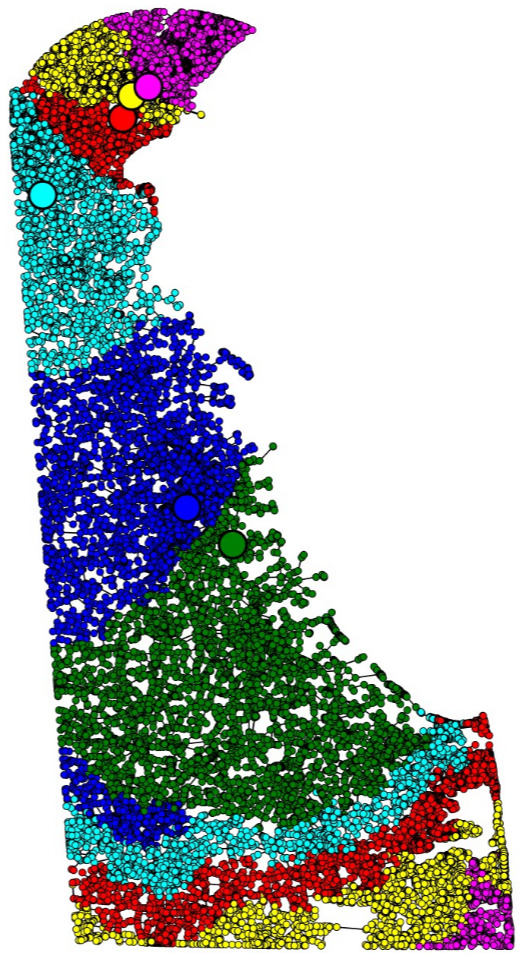} \\
Florida ($n=1037K,m=1315K$) & Delaware ($n=49K,m=60K$)
\end{tabular}
\caption{The solutions to 
the \emph{stable graph matching} problem for the 2010 road network of six US states, from the DIMACS database~\cite{DIMACS}. They consists of primary and secondary roads in the biggest connected component of the road networks. In each case, $n$ and $m$ denote the number of nodes and edges, respectively, and there are $k=6$ random centers with equal quota $n/k$.}
\label{fig:samplemaps2arxiv}
\end{figure*}

\begin{figure*}
\begin{tabular}{cc}
  \includegraphics[width=0.38\linewidth]{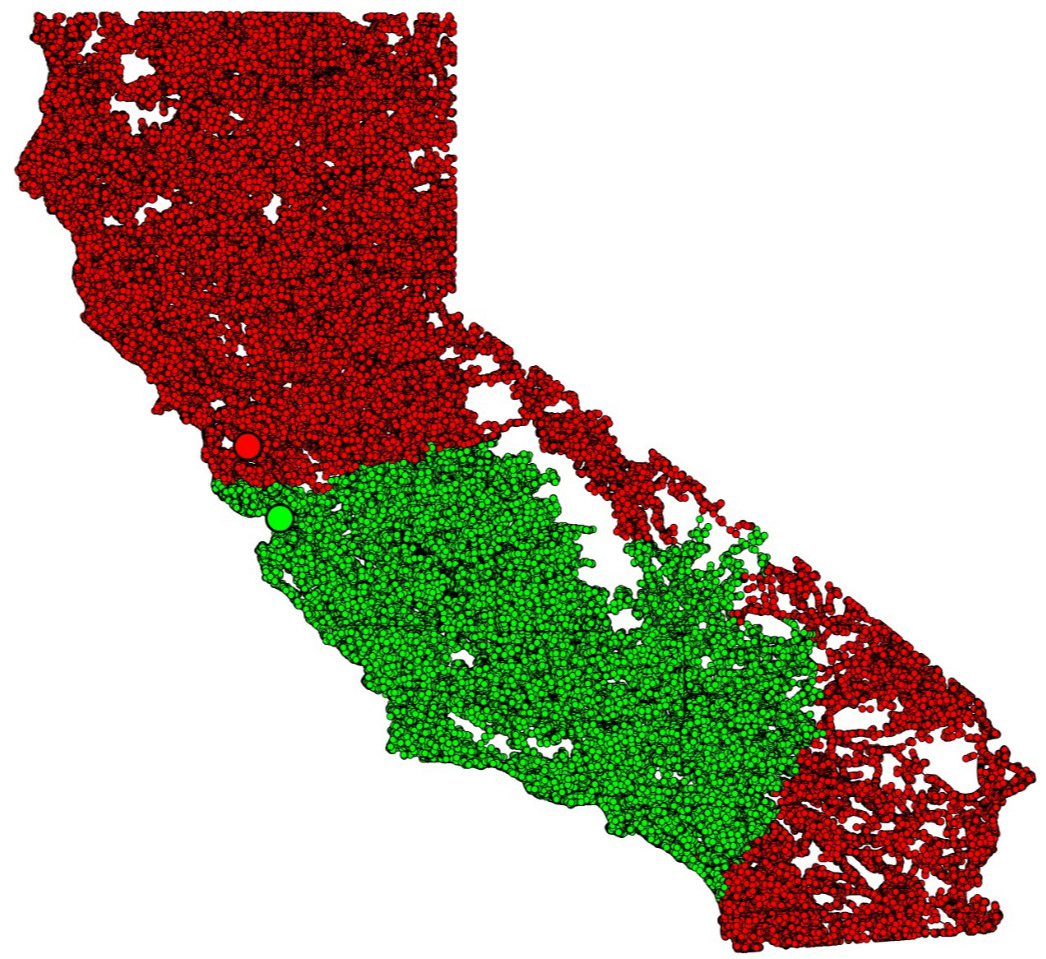} &   \includegraphics[width=0.38\linewidth]{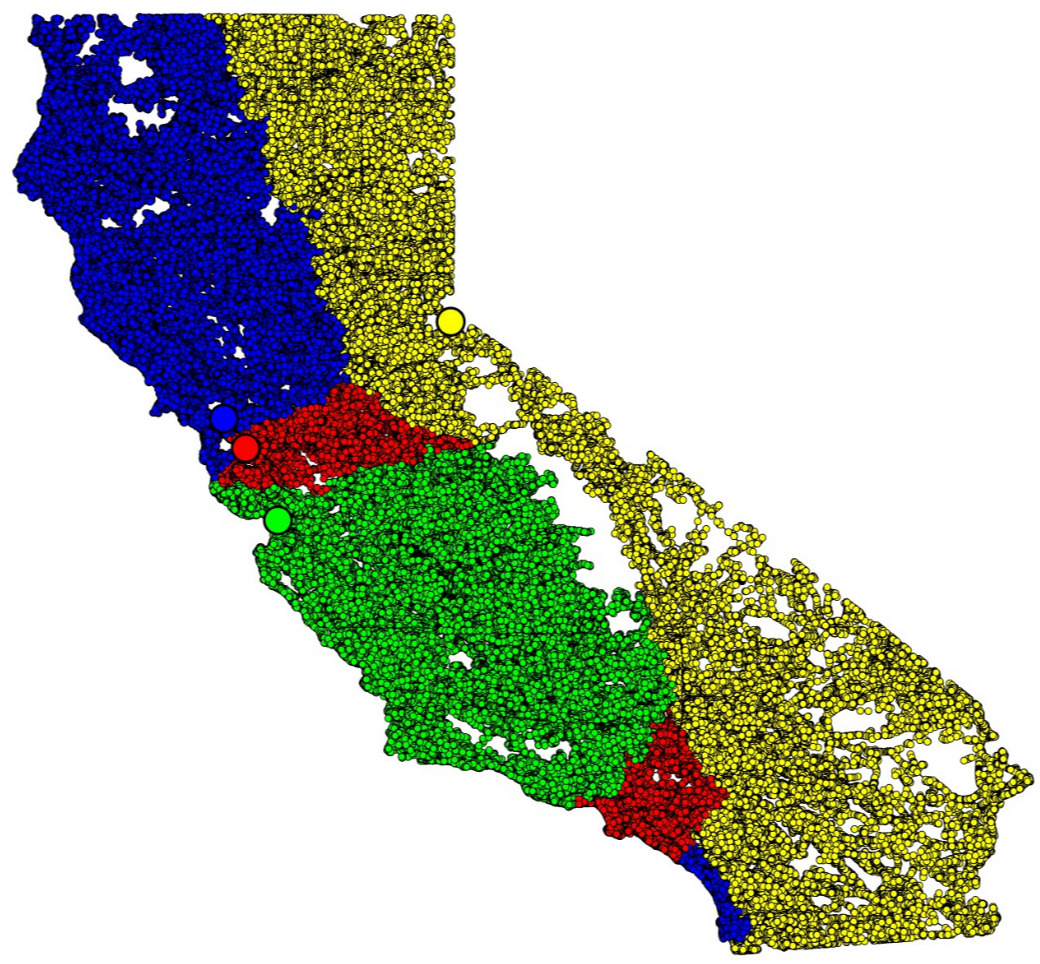} \\    
  $k=2$ & $k=4$ \\[6pt]
  \includegraphics[width=0.38\linewidth]{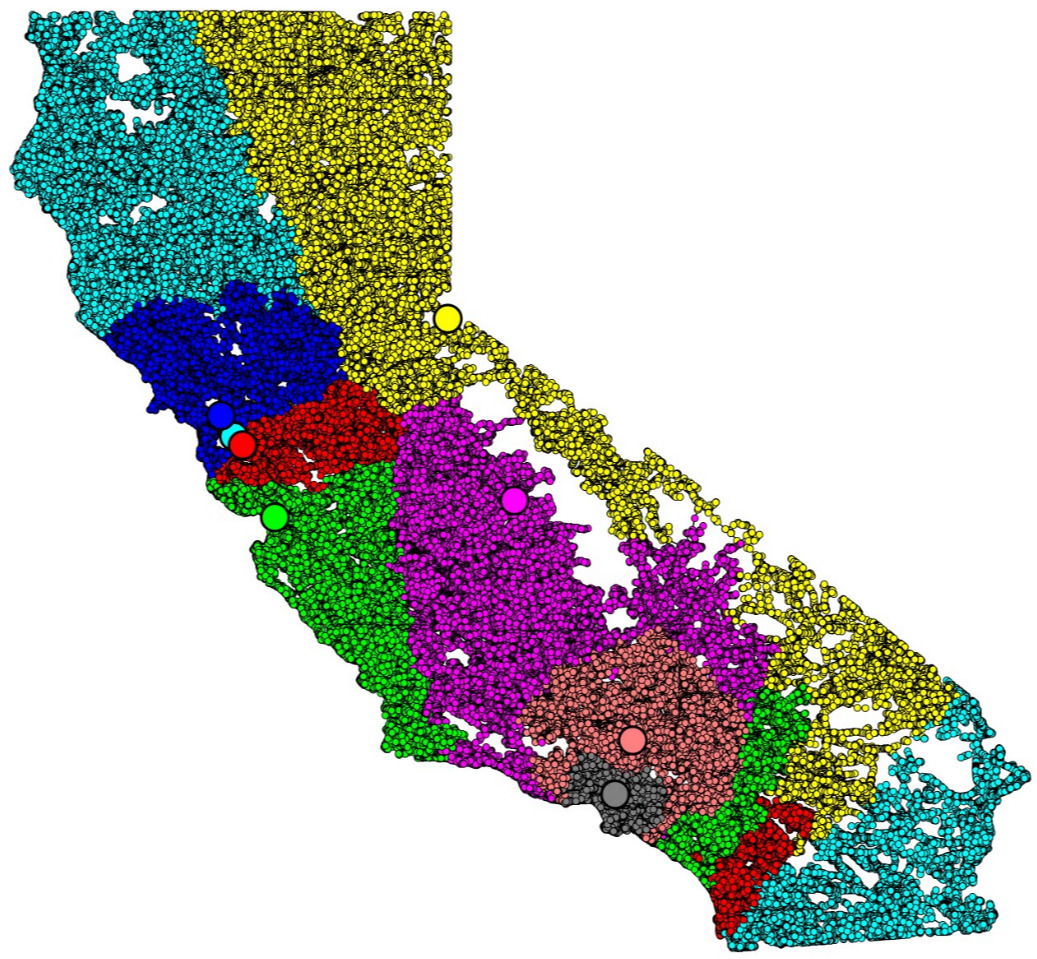} &   \includegraphics[width=0.38\linewidth]{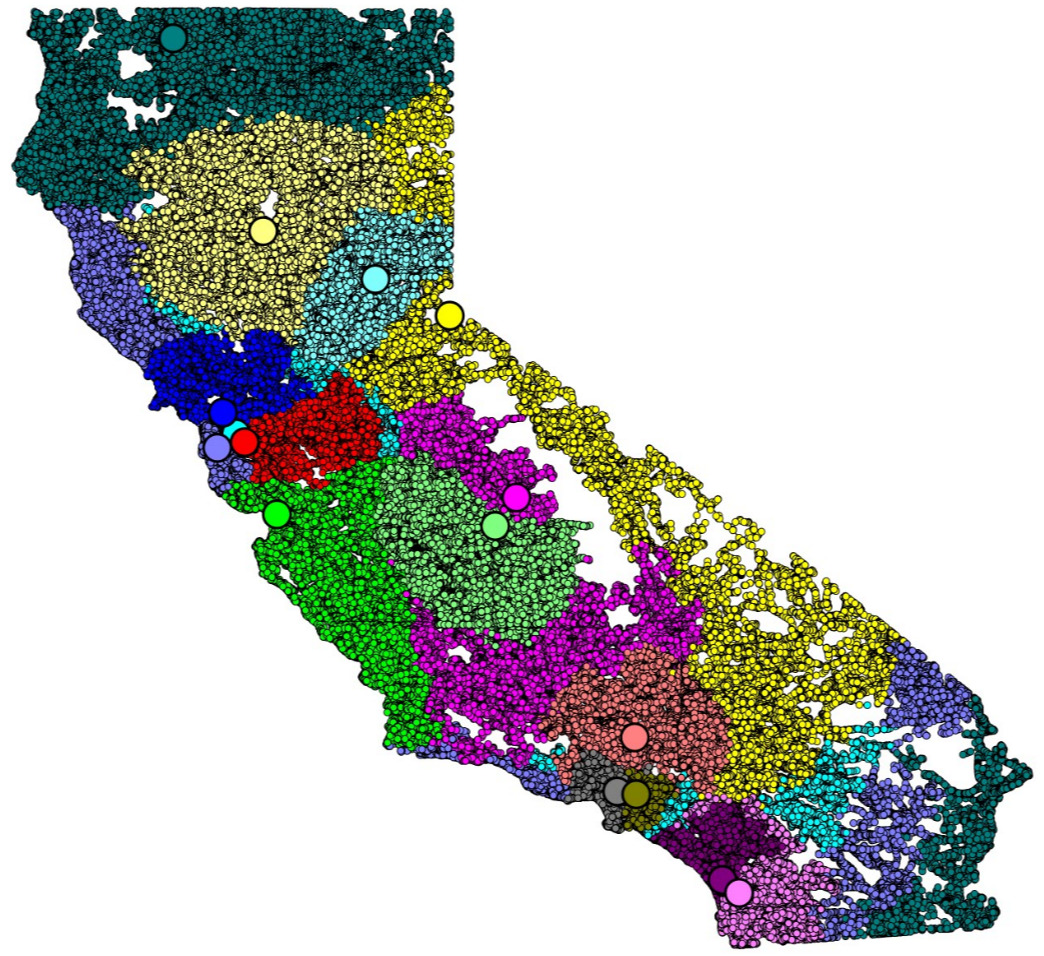} \\    
    $k=8$ & $k=16$ \\[6pt]
 \includegraphics[width=0.38\linewidth]{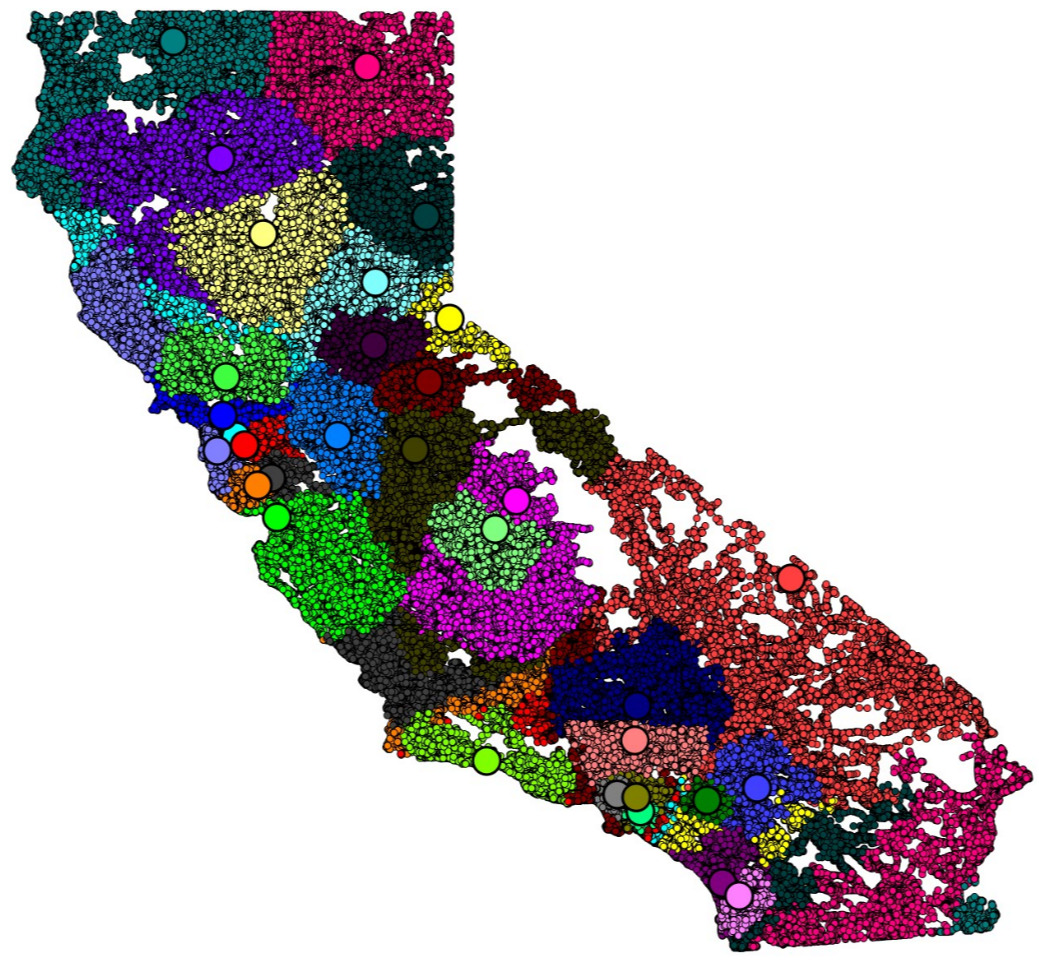} &   \includegraphics[width=0.38\linewidth]{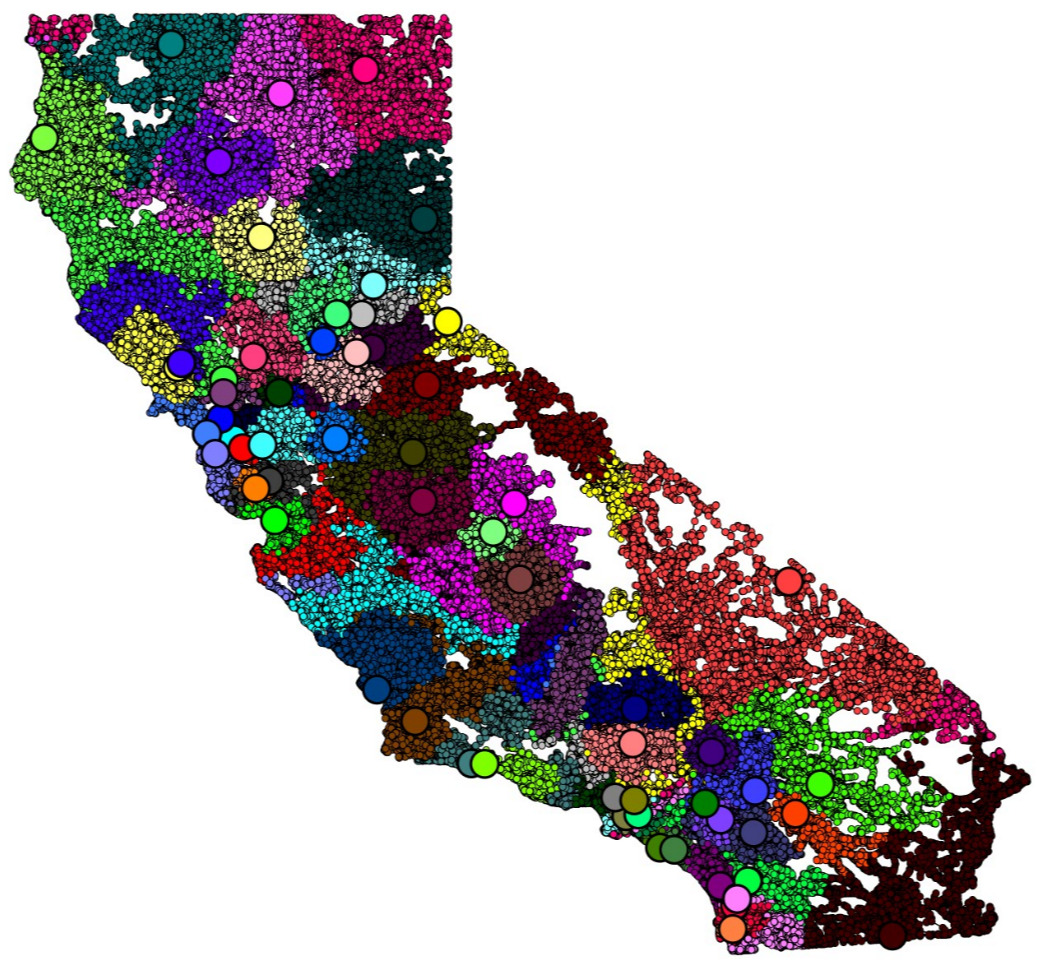}\\    
   $k=32$ & $k=64$ 
\end{tabular}
\caption{The solutions to 
the \emph{stable graph matching} problem for the 2010 road network of California from the DIMACS database~\cite{DIMACS} and different numbers of centers. The road network consists of primary and secondary roads in the biggest connected component, for a total of $n=1596K$ nodes and $m=1971K$ edges. The centers have been chosen randomly, and in each case they have equal quota $n/k$.}
\label{fig:diffkarxiv}
\end{figure*}

\end{document}